\documentclass[reqno,10.5pt]{amsart}
\usepackage{amsaddr}
\usepackage{hyperref}
\usepackage[margin=3cm]{geometry}
\usepackage{shortcuts}
\usepackage{complexity}
\usepackage{multirow}
\usepackage{array}
\usepackage{tabularx}
\usepackage{tikz}
\usepackage[latin1]{inputenc}
\usetikzlibrary{shapes,arrows}
\usetikzlibrary{matrix}
\usepackage{cite}
\input{Qcircuit}

\def\AbsSAT{\ComplexityFont{AbsSAT}}

\def\QC{\ComplexityFont{QC}}
\def\strong{\ComplexityFont{STR}}
\def\weak{\ComplexityFont{WEAK}} 
\def\diag{\textrm{diag}}
\def\CX{\textrm{CX}}

\def\postcnu{\ComplexityFont{post}\mathcal{C}_\nu\ComplexityFont{P}}
\def\postBQP{\ComplexityFont{postBQP}}
\def\postBPP{\ComplexityFont{postBPP}}
\def\Tof{\textrm{Tof}}
\def\anc#1{|\vec #1 \rangle_A}
\newcolumntype{L}[1]{>{\raggedright\let\newline\\\arraybackslash\hspace{0pt}}m{#1}}
\newcolumntype{C}[1]{>{\centering\let\newline\\\arraybackslash\hspace{0pt}}m{#1}}
\newcolumntype{R}[1]{>{\raggedleft\let\newline\\\arraybackslash\hspace{0pt}}m{#1}}

\begin{document}
\title{Further extensions of Clifford circuits and their classical simulation complexities}
\author{Dax Enshan Koh}
\address{Department of Mathematics, Massachusetts Institute of Technology, Cambridge, Massachusetts 02139, USA}
\email{daxkoh@mit.edu}

\begin{abstract}
Extended Clifford circuits straddle the boundary between classical and quantum computational power. Whether such circuits are efficiently classically simulable seems to depend delicately on the ingredients of the circuits. While some combinations of ingredients lead to efficiently classically simulable circuits, other combinations, which might just be slightly different, lead to circuits which are likely not. We extend the results of Jozsa and Van den Nest [Quant. Info. Comput. 14, 633 (2014)] by studying two further extensions of Clifford circuits. First, we consider how the classical simulation complexity changes when we allow for more general measurements. Second, we investigate different notions of what it means to `classically simulate' a quantum circuit. These further extensions give us 24 new combinations of ingredients compared to Jozsa and Van den Nest, and we give a complete classification of their classical simulation complexities. Our results provide more examples where seemingly modest changes to the ingredients of Clifford circuits lead to ``large" changes in the classical simulation complexities of the circuits, and also include new examples of extended Clifford circuits that exhibit ``quantum supremacy", in the sense that it is not possible to efficiently classically sample from the output distributions of such circuits, unless the polynomial hierarchy collapses.

\end{abstract}

\maketitle

\renewcommand\labelitemi{\tiny$\bullet$}
\section{Introduction}

Clifford circuits are an important class of circuits in quantum computation \cite{Gottesman98}. They have found numerous applications in quantum error correction \cite{GottesmanThesis}, measurement-based quantum computation \cite{Raussendorf, RaussendorfBriegel} as well as quantum foundations \cite{Pusey, Spekkens, penney2016quantum}. One of the central results about Clifford circuits is the Gottesman-Knill Theorem \cite{Gottesman98}, which states that such circuits can be efficiently simulated on a classical computer, and hence do not provide a speedup over classical computation. But this is known to be true only in a restricted setting -- whether or not we can efficiently classically simulate such circuits seems to depend delicately on the `ingredients' of the circuit, for example, on the types of inputs we allow, whether or not intermediate measurements are adaptive, the number of output lines, and even on the precise notion of what it means to \textit{classically simulate} a circuit. These cases were considered by Jozsa and Van den Nest \cite{Jozsa}, who showed that many of these `extended' Clifford circuits are in fact not classically simulable under plausible complexity assumptions.

One of the main motivations for studying extended Clifford circuits is that they shed light on the relationship between quantum and classical computational power. Are quantum computers more powerful than their classical counterparts? If so, what is the precise boundary between their powers? One approach to answering this question is to consider restricted models of quantum computation and study their classical simulation complexities, i.e. how hard it is to classically simulate them. For example, suppose that we start with a restricted model that is efficiently classically simulable. If adding certain ingredients to the restricted model creates a new class that is universal for quantum computation, then we could regard those ingredients as an essential `resource' for quantum computational power \cite{Jozsa}. Extended Clifford circuits, as a restricted model of quantum computation, are especially well-suited for this approach as they straddle the boundary between classical and quantum computational power. One could give many examples where adding a seemingly modest ingredient to an extended Clifford circuit changes it from being efficiently classically simulable to one that is likely not.

Understanding how the classical simulation complexities of extended Clifford circuits change when various ingredients are added is a central goal of this paper. In \cite{Jozsa}, Jozsa and Van den Nest tabulate the classical simulation complexities of extended Clifford circuits with 16 different combinations of ingredients. In particular, they consider the different combinations of ingredients that arise from 4 binary choices: computational basis inputs vs product state inputs, single-line outputs vs multiple-line outputs, nonadaptive measurements vs adaptive measurements, and weak vs strong simulation. They show that the classical simulation complexities of the extended Clifford circuits are of 4 different types (we use slightly different terminology here): (i) $\P$, which means that the circuits can be efficiently simulated classically, (ii) $\QC$, which means that the circuits are universal for quantum computation, (iii) $\#\P$, which means that the problem of classically simulating the circuits is a $\#\P$-hard problem, and (iv) $\PH$, which means that if the circuits are efficiently classically simulable, then the polynomial hierarchy collapses. 

In this paper, we extend the results in \cite{Jozsa} in two different ways. First, we study how the classical simulation complexity changes when we employ a weaker notion of simulation than strong simulation, which we call $\strong(n)$ simulation (short for strong-$n$ simulation). While strong simulation requires that the joint probability as well as any marginal probabilities be computed, in $\strong(n)$ simulation, we require only that the joint probability be computed. Note that such a notion seems incomparable with weak simulation. Second, we study how the classical simulation complexity changes when we allow for general product measurements (called OUT(PROD)) instead of just the computational basis measurements (called OUT(BITS)) that were considered in \cite{Jozsa}. With these additional ingredients, the number of different combinations of ingredients grows to 40. In Table \ref{mainTable}, we tabulate the classical simulation complexities of each of these cases.

We now make a few remarks about the extended Clifford circuits labeled in Table \ref{mainTable} by $\PH$. These are examples of `intermediate'  or restricted quantum circuit models which are not believed to be universal for quantum computation (or perhaps even classical computation), but which exhibit a form of `quantum supremacy' \cite{preskill2012quantum,aaronsonchen}, in the sense that they can sample from distributions that are impossible to sample from classically, unless the plausible complexity assumption of the polynomial hierarchy being infinite is false. In \cite{Jozsa}, Jozsa and Van den Nest give an example of such a circuit model: nonadaptive Clifford circuits with product state inputs and computational basis measurements. In this paper, we show that the same behavior holds if we restricted our circuits to having computational basis inputs but allowed them to have arbitrary single-qubit measurements performed at the end of the circuit (see Theorem \ref{Thm3}). Note that a similar `quantum supremacy' behavior can be found in several other restricted quantum models, like boson sampling \cite{AaronsonArkhipov}, IQP circuits \cite{bremner2010classical, bremner2015average, bremner2016achieving}, QAOA \cite{farhi2016quantum}, DQC1 circuits \cite{morimae2014hardness}, and a model based on integrable quantum theories in 1+1 dimensions \cite{mehraban2015computational, aaronson2016computational}.

The rest of the paper is structured as follows. In Section \ref{sec:prelim}, we introduce the Clifford circuit model and discuss various extensions of Clifford circuits. In Section \ref{sec:Notions}, we define different notions of classical simulation of quantum computation. In Section \ref{sec:results}, we summarize our main results in the form of Table \ref{mainTable} and discuss some implications of our results. Our main theorems are Theorems \ref{Thm1}--\ref{Thm6}, whose proofs are presented in Appendices \ref{sec:thm1}--\ref{sec:thm6}.

\section{Preliminary definitions and notations}
\label{sec:prelim}

We review the definitions introduced in \cite{Jozsa}: the standard Pauli matrices are denoted by $I, X, Y, Z$, and an $n$-qubit \textit{Pauli operator} is defined to be any operator of the form $P= i^k P_1 \otimes \ldots \otimes P_n$, where $k=0,1,2,3$ and each $P_i$ is a Pauli matrix. The set of $n$-qubit Pauli operators forms a group $\mathcal P_n$, called the \textit{Pauli group}. The $n$-qubit \textit{Clifford group} $\mathcal C_n$ is defined to be the normalizer of the Pauli group $\mathcal P_n$ in the $n$-qubit unitary group $\mathcal U_n$, i.e.\@ $\mathcal C_n = \{U \in \mathcal U_n| U \mathcal P_n U^\dag = \mathcal P_n\}$. Elements of the Clifford group are called \textit{Clifford operations}. Clifford operations have an alternative characterization \cite{NielsenChuang}: an $n$-qubit operator $C$ is a Clifford operation if and only if it can be written as a circuit consisting of $O(n^2)$ gates from the following list: the Hadamard gate $H = 1/\sqrt 2(X+Z)$, the phase gate $S=\diag(1,i)$, and the CNOT gate $\CX_{ab}  = \ket 0\bra 0_a \otimes I_b + \ket 1\bra 1_a \otimes X_b$. Following the terminology in \cite{Jozsa}, we call these gates the \textit{basic Clifford gates}. A \textit{unitary Clifford circuit} is one that comprises only the basic Clifford gates. A \textit{Clifford circuit} is one that consists of not just the basic Clifford gates but also single-qubit intermediate measurement gates in the computational basis. 

We consider \textit{Clifford computational tasks} of the following form:
\begin{enumerate}
\item Start with an $n$-qubit pure input state $\ket{\psi_{\mathrm{in}}}$.
\item Apply to $\ket{\psi_{\mathrm{in}}}$ a Clifford circuit $B$, which may be expressed as:
\begin{eqnarray}
\label{eq:adaptiveCircuit}
B(x_1,\ldots, x_K)&=& C_K(x_1,\ldots, x_K) M_{i_K (x_1 ,\ldots, x_{K-1})} (x_K) \ldots \nonumber\\ && \qquad C_2(x_1,x_2) M_{i_2 (x_1)} (x_2) C_1 (x_1 ) M_{i_1} (x_1) C_0,
\end{eqnarray}
where each $C_i(x_1,\ldots, x_i)$ is a unitary Clifford circuit and $M_i(x)$ indicates a measurement on qubit line $i$ with measurement result $x$. In general, $B$ is taken to be an adaptive circuit, i.e.\@ the $i$th unitary Clifford circuit $C_i$ depends on previous measurement results $x_1,\ldots, x_i$. Let $N$ denote the total number of gates in $B$. Assume that there are no extraneous qubits, so that $n=O(N)$.
\item Measure all $n$ qubit lines using a projection-valued measure $\{|\beta_{y_1,\ldots, y_n}\rangle\langle \beta_{y_1,\ldots, y_n}|\}_{y_1,\ldots, y_n}$, with measurement outcome $y_1 \ldots y_n \in \{0,1\}^n$.
\end{enumerate}

In this work, we restrict our attention to product state inputs and product measurements (i.e. arbitrary single-qubit measurements), i.e. 
\begin{enumerate}
\item Inputs are $\ket{\psi_{\mathrm{in}}}=\ket{\alpha_1}\ket{\alpha_2}\ldots \ket{\alpha_n}$, where each $ \ket{\alpha_i} \in \mathbb C^2$.
\item Measurements directions are $\ket{\beta_{y_1,\ldots y_n}} = \ket{\beta_1^{y_1}}\ket{\beta_2^{y_2}}\ldots \ket{\beta_n^{y_n}}$, where each  $ \ket{\beta_i^{y_i}}\in \mathbb C^2$.
\end{enumerate}

Note that for each $i$, by completeness, $\ket{\beta_i^0}\bra{\beta_i^0} +\ket{\beta_i^1}\bra{\beta_i^1}=I $. Hence, we need to just specify $\{\ket{\beta_i^0}\bra{\beta_i^0}\}_i$ in order to completely specify the product measurement. A description of the Clifford computational task is thus given by the three-tuple
\begin{equation}
\label{CliffDesc}
T=(\ket{\alpha},B,\ket{\beta}),
\end{equation}
where $\ket\alpha=\ket{\alpha_1}\ket{\alpha_2}\ldots \ket{\alpha_n}$ is the initial state, $B$ is the description of the Clifford circuit, and $\ket{\beta}= \ket{\beta_1^0}\ket{\beta_2^0}\ldots \ket{\beta_n^0}$ are the measurement directions.

Now, each product state input can be seen as arising from applying a product unitary to the computational basis states, i.e.\@ there exist single-qubit unitary operators $V_1,\ldots, V_n$ such that $V_1\otimes\ldots\otimes V_n \ket {0\ldots 0} = \ket\alpha$. Likewise, every product measurement operator can be seen as arising from applying a product unitary operator followed by measuring in the computational basis. More precisely, a measurement in the direction  $\ket{\beta_1^{y_1}}\ket{\beta_2^{y_2}}\ldots \ket{\beta_n^{y_n}}$ is equivalent to the application of a unitary operator $U_1^\dag\otimes\ldots\otimes U_n^\dag$ followed by a measurement in the computational basis, where the $y_i$th (with zero indexing) column of $U_i$ is given by $U_i \ket{y_i} = \ket{\beta_i^{y_i}}$. 

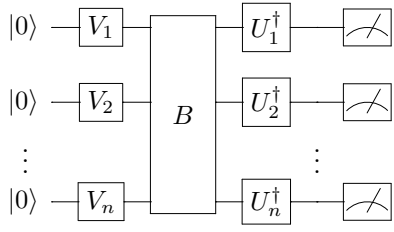
\begin{figure}
\begin{align}
\Qcircuit @C=1em @R=1em {
& &\lstick{\ket{0}} & \gate{V_1}  & \multigate{3}{ \ B \ } & \gate {U_1^\dag} & \qw & \meter \\ 
& &\lstick{\ket{0}} & \gate{V_2}  & \ghost{ \ B \ } & \gate {U_2^\dag} & \qw & \meter \\ 
& \vdots & & &  &   & \vdots  \\
& &\lstick{\ket{0}} & \gate{V_n}  & \ghost{ \ B \ } & \gate {U_n^\dag} & \qw & \meter \\ 
} \nonumber
\end{align}
\caption{\small Circuit diagram for the Clifford computational tasks considered in this paper. The gates $V_i$ and $U_i^\dag$ are arbitrary single qubit unitaries, $B$ is a Clifford circuit, the input state is the all-zero computational basis state, and the output measurement is performed in the computational basis. }
\label{CliffordWithUnitaries}
\end{figure}

Hence, the Clifford computational tasks we consider are of the structure shown in Figure \ref{CliffordWithUnitaries}. They may alternatively be represented by the 3-tuple
\begin{equation} 
\label{CliffDesc2}
T=(\{V_i\}_{i=1}^n,B,\{U_i\}_{i=1}^n).
\end{equation}
We will use the above two descriptions in Eqs.\,\eqref{CliffDesc} and \eqref{CliffDesc2} of Clifford tasks interchangeably, and even allow for mixed descriptions, for example, $T=(\ket{\alpha},B,\{U_i\}_{i=1}^n)$.

We'll now write down expressions for the probabilities of outcomes. For a computational task $T=(\ket{\alpha},B,\ket{\beta})$ and subset $I=\{i_1,\ldots,i_s\} \subseteq [n]$, let $P_T^I(y_{i_1},\ldots,y_{i_s})$ be the marginal probability that the outputs $y_{i_1},\ldots,y_{i_s}$ are obtained in the lines $i_1,\ldots,i_s$. Define $P_T(y_1,\ldots,y_n) = P_T^{[n]}(y_1,\ldots,y_n)$ to be the probability of the outcome $y_1 \ldots y_n$.

For the adaptive circuit described by Eq.\,\eqref{eq:adaptiveCircuit}, if the intermediate measurement results are $x_1\ldots x_K$, then the  density operator of the final state is given by $ B(x_1,\ldots, x_K) [\rho_\alpha]$, where $\rho_\alpha = \ket \alpha \bra \alpha$. We use the notation $C[\rho]$ to denote the state that is obtained when we apply $C$ to the density matrix $\rho$, i.e.\@ $C[\rho] = C\rho C^\dag$.
The probability that the result $x_1\ldots x_K$ occurs is given by $$p(x_1,\ldots,x_K) = p(x_K|x_1,\ldots,x_{K-1})  p(x_{K-1}|x_1,\ldots,x_{K-2})\ldots p(x_2|x_1)p(x_1),$$ where 
\begin{eqnarray}
p(x_j|x_1,\ldots, x_{j-1}) &=& \tr\{ \ket{x_j}\bra{x_j}_{i_j} C_{j-1}(x_1,\ldots, x_{j-1}) M_{i_{j-1} (x_1 ,\ldots, x_{j-2})} (x_{j-1}) \ldots  \nonumber\\ &&\times C_1 (x_1 ) M_{i_1} (x_1) C_0 [\rho_\alpha]\} .
\end{eqnarray}

The final output state is then given by
$$B[\rho_\alpha] = \sum_{x_1\ldots x_K} p(x_1,\ldots, x_K)B(x_1,\ldots, x_K)  [\rho_\alpha]. $$

Hence, the outcome probabilities are given by
$$p_T(y_1,\ldots, y_n) = \langle \beta_{y_1,\ldots, y_n}| B[\rho_\alpha]|\beta_{y_1,\ldots, y_n}\rangle,$$
and the marginal probabilities are given by
\begin{equation}
\label{marginals}
p_T^{I}(y_{i_1},\ldots, y_{i_s}) = \sum_{y_{k_1}\ldots y_{k_{n-s}}} p_T(y_1,\ldots, y_n),
\end{equation}
where $\{k_1,\ldots,k_{n-s}\} = [n]-I$.

We consider the following 3 binary choices of ingredients: 
\begin{enumerate}
\item Inputs: IN(BITS) vs IN(PROD)
\item Intermediate measurements: NONADAPT vs ADAPT
\item Outputs: OUT(BITS) vs OUT(PROD)
\end{enumerate}

The first two cases have been considered in \cite{Jozsa}: IN(BITS) and IN(PROD) refer to having computational basis inputs and product state inputs respectively, while NONADAPT and ADAPT refer to nonadaptive and adaptive measurements respectively. Note that in \cite{Jozsa}, all the output measurements are performed in the computational basis (call this case OUT(BITS)). In this paper, we study how the classical simulation complexity changes when we allow for more general measurements. For the sake of symmetry with the inputs, we introduce the new ingredient OUT(PROD), which refers to product measurements, i.e.\@ when the $U_i$'s in \eq{CliffDesc2} are unrestricted. Note that we allow product measurements only at the output -- intermediate measurements are always single-qubit measurements in the computational basis.

These 3 binary choices lead to $2^3=8$ different subsets of Clifford computational tasks. Let $\nu \in$ \{(IN(BITS), NONADAPT, OUT(BITS)), (IN(BITS), NONADAPT, OUT(PROD)), $\ldots$ \} be one of these 8 subsets. We shall denote the subset of Clifford computational tasks corresponding to $\nu$ by $\mathcal C_\nu$. Note that unlike \cite{Jozsa}, we do not include OUT(1) and OUT(MANY) as ingredients in our circuit. Instead, we assume without loss of generality that all $n$ qubit lines are measured. This is justified by the principle of implicit measurement, which states that any unterminated quantum wires at the end of the circuit can be assumed to be measured \cite{NielsenChuang}. The number of output lines we simulate will be specified by the notion of simulation instead. We discuss various notions of simulation in the next section.

\section{Notions of classical simulation}
\label{sec:Notions}

In \cite{Jozsa}, Jozsa and Van den Nest consider two notions of classical simulation, namely weak ($\weak$) and strong ($\strong$) simulation. 	A weak simulation involves providing a sample of the output distribution, whereas a strong simulation involves calculating the joint output probabilities as well as the marginal probabilities. Neither of these definitions places a restriction on the number of output registers to be simulated. To take this into account, we shall introduce finer-grained notions of simulation, namely $\strong(f(n))$ and $\weak(f(n))$ (short for strong-$f(n)$ and weak-$f(n)$) simulation.

Let $f(n)$ be either the constant function $f(n)=1$ or the identity function $f(n)=n$ (in this paper, we restrict our attention to these cases, though one might certainly consider other functions $f$, like $f(n) = \log (n)$).
\begin{definition} ($\strong(f(n)$))
A $\strong(f(n))$ simulation of a subset of Clifford computational tasks $\mathcal C_\nu$ is a deterministic classical algorithm that on input $\langle T, I, y\rangle$, where $T\in \mathcal C_\nu$ is a task on $n$ qubits, $I = \{i_1,\ldots, i_{f(n)}\} \subseteq [n]$ and $y_I = \{y_{i_1}, \ldots, y_{i_{f(n)}}\}$, outputs $p_T^{I} (y_{i_1},\ldots, y_{i_{f(n)}})$.
\end{definition}

\begin{definition} ($\weak(f(n)$))
A $\weak(f(n))$ simulation of a subset of Clifford computational tasks $\mathcal C_\nu$ is a randomized classical algorithm that on input $\langle T, I\rangle$, where $T\in \mathcal C_\nu$ is a task on $n$ qubits and $I = \{i_1,\ldots, i_{f(n)}\} \subseteq [n]$, outputs $y_{i_1},\ldots, y_{i_{f(n)}}$, with probability $p_T^{I} (y_{i_1},\ldots, y_{i_{f(n)}})$. 
\end{definition}

\tikzstyle{block} = [rectangle, 
    text width=5em, text centered, minimum height=3em]
\tikzstyle{line} = [draw, -latex']

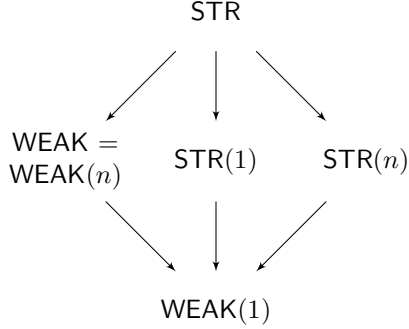
\begin{figure}
\begin{tikzpicture}[node distance = 2cm, auto]
    \node [block] (str) {$\strong$};
    \node [block, below of=str] (str1){$\strong(1)$};
    \node [block, left of=str1] (wk){$\weak = \weak(n)$};
    \node [block, right of=str1] (strn){$\strong(n)$};
    \node [block, below of=str1] (wk1){$\weak(1)$};
    
     \path [line] (str) -- (str1);
     \path [line] (str) -- (wk);
     \path [line] (str) -- (strn);
     \path [line] (str1) -- (wk1);
     \path [line] (wk) -- (wk1);
     \path [line] (strn) -- (wk1);

\end{tikzpicture}
\caption{Relationships between different notions of classical simulation of Clifford computational tasks. An arrow from $A$ to $B$ ($A \rightarrow B$) means that an efficient $A$-simulation of a computational task implies that there is an efficient $B$-simulation for the same task. The statement $\weak = \weak(n)$ is shorthand for $\weak \rightarrow \weak(n)$ and $\weak(n) \rightarrow \weak$.}
\label{fig:notionsOfxSimulation}
\end{figure}

$\strong$ and $\weak$ simulations are defined in exactly the same way, except that we place no restrictions on the size of the subset of output lines $|I|$ in the simulation. Note that this agrees with the definitions of strong and weak simulations in \cite{Jozsa}.

Let $\mathcal S \in \{ \strong(n),\strong(1), \strong, \weak(n), \weak(1), \weak \}$ be one of these 6 notions of simulation. We define an $\mathcal S$-simulation of a subset of Clifford computational tasks $\mathcal C_\nu$ to be \textit{efficient} if the simulation runs in $\poly(N)$-time, where $N$ is the number of gates in the $\mathcal C_\nu$-circuit. Let $\P\mathcal S$ be the set of all tasks $\mathcal C_\nu$ that have an efficient $\mathcal S$-simulation.

An immediate observation is that $\P\weak=\P\weak(n)$. The forward inclusion holds by definition, and the backward inclusion holds because we could sample from any subset $I$ by just sampling from all $n$ lines and ignoring the qubit lines that are not in $I$. From their definitions, we also immediately get the following inclusions: $\P\strong \subseteq \P\strong(1)$, $\P\strong \subseteq \P\strong(n)$ and $\P\weak \subseteq \P\weak(1)$. How does weak simulation compare with strong simulation? From Proposition 1 of \cite{Terhal}, it follows that $\P\strong \subseteq \P\weak$, $\P\strong(1) \subseteq \P\weak(1)$ and $\P\strong(n) \subseteq \P\weak(1)$. Note that the notions $\P\strong(n)$ and $\P\strong(1)$ are in general incomparable -- the forward inclusion ($\P\strong(n) \subseteq \P\strong(1)$) does not hold in general because computing a marginal distribution directly from the joint distribution involves summing an exponential number of terms and cannot be performed efficiently unless there is some structure in the problem. The backward inclusion ($\P\strong(n) \supseteq \P\strong(1)$) does not hold in general because knowing just the marginal distributions does not allow us to infer the joint distribution. We summarize the relationships between the different notions of simulation stated above in Figure \ref{fig:notionsOfxSimulation}.

\section{Results and discussion}
\label{sec:results}

\begin{table}
\begin{center}
  \renewcommand{\arraystretch}{1.4}
\begin{tabular}{|C{1.2cm}|C{1.2cm}|C{1.2cm}||C{1.6cm}|C{1.6cm}||C{1.6cm}|C{1.6cm}|C{1.6cm}|}
\hline
\multicolumn{3}{|c||}{}  &  \multicolumn{2}{c||}{Weak} & \multicolumn{3}{c|}{Strong} \\ \cline{4-8}
\multicolumn{3}{|c||}{} &  \weak(1)& \weak($n$) & \strong(1) & \strong($n$) & \strong \\ \hline
\multirow{6}{2cm}{ OUT\\(BITS) } & \multirow{2}{2cm}{\shortstack{ NON-\\ADAPT}} &  \shortstack{\\IN\\(BITS)} & \P \newline \footnotesize (i)  & \P \newline \footnotesize (ii) & \P \newline \footnotesize (iii) & \P \newline \footnotesize (iv) & \boxed \P  \newline \footnotesize (JV4) \\ \cline{3-8}
& &  IN \newline (PROD) & \P \newline \footnotesize (v) & \boxed\PH  \newline \footnotesize (JV7) & \P \newline \footnotesize (JV1) & \boxed{\#\P} \newline \footnotesize (Thm \ref{Thm1}) & \#\P \newline \footnotesize (JV6) \\ \cline{2-8} 
& \multirow{2}{2cm}{ADAPT} &  \shortstack{IN\\(BITS)} & \P \newline \footnotesize  (vi)  & \boxed\P \newline \footnotesize (JV5) & \boxed{\#\P} \newline \footnotesize (JV2) & \boxed{\#\P} \newline \footnotesize (Thm \ref{Thm2}) & \#\P \newline \footnotesize (vii) \\ \cline{3-8}
& &  IN \newline  (PROD) & \boxed{\QC} \newline \footnotesize (JV3) & \QC \newline \footnotesize (viii) & \#\P \newline \footnotesize (ix)  & \#\P \newline \footnotesize (x)  & \#\P \newline \footnotesize (xi)  \\ \hline \hline
\multirow{6}{2cm}{ OUT\\(PROD) } &
\multirow{2}{2cm}{\shortstack{ NON-\\ADAPT}} &  \shortstack{\\IN\\(BITS)} & \P \newline \footnotesize (xii)  & \boxed\PH \newline \footnotesize (Thm \ref{Thm3}) & \P \newline \footnotesize (xiii) & \boxed{\#\P} \newline \footnotesize (Thm \ref{Thm4}) & \#\P \newline \footnotesize (xiv)   \\  \cline{3-8}
& &  IN \newline (PROD) & \P \newline \footnotesize (xv)  & \PH \newline \footnotesize (xvi)  & \boxed\P \newline \footnotesize (Thm \ref{Thm5}) & \#\P \newline \footnotesize (xvii)  & \#\P \newline \footnotesize (xviii)   \\ \cline{2-8} 
& \multirow{2}{2cm}{ADAPT} &  \shortstack{IN\\(BITS)} & \boxed\P \newline \footnotesize (Thm \ref{Thm6}) & \PH \newline \footnotesize (xix) & \#\P \newline \footnotesize (xx) & \#\P \newline \footnotesize (xxi) & \#\P \newline \footnotesize (xxii) \\  \cline{3-8}
& &  IN \newline (PROD) & \QC \newline \footnotesize (xxiii) & \QC \newline \footnotesize (xxiv) & \#\P\newline \footnotesize (xxv) & \#\P\newline \footnotesize (xxvi) & \#\P\newline \footnotesize (xxvii)   \\ \hline
\end{tabular}
\end{center}
\caption{\small Classification of the classical simulation complexities of families of Clifford circuits with different ingredients. $\P$ stands for efficiently classically simulable. $\#\P$ stands for $\#\P$-hard. $\QC$ stands for $\QC$-hard and $\PH$ stands for ``if efficiently classically simulable, then the polynomial hierarchy collapses". The proofs of JV 1--7 can be found in \cite{Jozsa}. Theorems \ref{Thm1}--\ref{Thm6} are about cases not found in \cite{Jozsa} and are the main results of this paper. (i)--(xxvii) are results that follow immediately from these theorems by using the rules in Appendix \ref{sec:SimplifyingRules}. The 11 cases with boxed symbols are the core theorems, from which all other cases can be deduced using rules which we describe in Appendix \ref{sec:SimplifyingRules}. These include all the main theorems JV 1--7 and Theorems \ref{Thm1}--\ref{Thm6}, except JV1 and JV6, which turn out to be special cases of Theorem \ref{Thm5} and Theorem \ref{Thm1} respectively. }
\label{mainTable}
\end{table}

In Section \ref{sec:prelim}, we introduced 3 binary choices of ingredients. In Section \ref{sec:Notions}, we described 5 different notions of classical simulation (see Figure \ref{fig:notionsOfxSimulation}). This gives a total of $2^3 \times 5 = 40$ different cases, whose classical simulation complexities we classify in Table \ref{mainTable}. The entries of the table should be understood as follows: for a subset of computational tasks $\mathcal C_\nu$, and a notion of simulation $\mathcal S$,
\begin{itemize} 
\item $\P$ (classically efficiently simulable) means that  $\mathcal C_\nu \in \P\mathcal S$.
\item $\#\P$ (which stands for $\#\P$-hard) means that an efficient $\mathcal S$-simulation of $\mathcal C_\nu$ would give rise to an efficient algorithm for the $\#\P$-complete problems.
\item $\QC$ (which stands for quantum-computing universal) means that $\mathcal C_\nu$ is universal for quantum computation.
\item $\PH$ means that an efficient $\mathcal S$-simulation of $\mathcal C_\nu$ would imply a collapse of the polynomial hierarchy.
\end{itemize}

Our main results are Theorems \ref{Thm1}--\ref{Thm6}, whose proofs we present in Appendices \ref{sec:thm1}--\ref{sec:thm6}. Using the rules in Appendix \ref{sec:SimplifyingRules}, these theorems, together with the results\footnote{JV = Jozsa and Van den Nest \cite{Jozsa}} JV 1--7 from \cite{Jozsa}, give a complete classification of the classical simulation complexities of all the 40 cases. 

A few remarks are in order. First, we note that the entries in the last two columns of Table \ref{mainTable} are identical. This means that even though the notions $\strong(n)$ and $\strong(1)$ seem to be incomparable, the former is not easier to perform than the latter for the Clifford computational tasks considered in this paper.  We note that Theorem \ref{Thm1}, which generalizes (JV6), implies that being able to compute only the joint probabilities already suffices in enabling us to solve the $\#\P$-hard problems: we do not require the full power of strong simulation for that.

Second, we note the symmetry between inputs and outputs: for example, the 2nd and 5th rows of Table \ref{mainTable} are identical, i.e.\@ the simulation complexity is the same whether product unitaries are applied at the beginning or at the end of the circuit. In particular, for (JV7), the key to collapsing the polynomial hierarchy was that the magic state $\ket{\pi/4} = 1/\sqrt{2}(\ket 0 + e^{i\pi/4}\ket 1)$ together with postselection can simulate the $T=\diag(1,e^{i\pi/4})$ gate. For Theorem \ref{Thm3}, although we did not have magic state inputs at our disposal, we still managed to get a similar result to (JV7) by showing that the $T$ gate can be simulated by arbitrary single-qubit measurements with postselection.

Third, we note that Theorem \ref{Thm5} is a generalization of JV1. In fact, a stronger result can similarly be shown to be true: for any constant $b$, there exists an efficient $\strong(b)$-simulation of circuits belonging to OUT(PROD), NONADAPT, IN(PROD). In \cite{AaronsonGottesman}, Aaronson and Gottesman present algorithms for simulating two separate classes of extended Clifford circuits: circuits with non-stabilizer initial states, and circuits with non-stabilizer gates. A consequence of their results is that it is efficient to simulate (in the $\strong(b)$-sense) nonadaptive tasks with either of the following ingredients: 1. product state inputs with computational basis measurements (which is the content of JV1). 2. computational basis inputs with product measurements (which is the content of case xiii) -- since this is equivalent to applying $b$ single-qubit gates just before a computational basis measurement. Theorem \ref{Thm5} is slightly more general than either of these cases. Essentially, it combines the case involving product state inputs and the case involving product measurements and shows that the new task is still in $\P\strong(b)$. 

\section{Concluding remarks}

We have demonstrated how the classical simulation complexities of extended Clifford circuits change when various ingredients in the circuits are varied. It would be interesting to study other ingredients of Clifford circuits as well, e.g., mixed input states \cite{AaronsonGottesman}, states (as well as transformations and measurements) with positive Wigner representations \cite{Eisert, Veitch}, and non-commutative extensions like XS-stabilizer states \cite{Ni}. Most of these extensions have previously been considered separately, and it will be fruitful to study the classical simulation complexities of computational tasks with these different combinations of ingredients.

Our discussion of extended Clifford circuits has been restricted to involve only exact simulation. But any classical simulation device in the real world is subject to noise and decoherence. Taking this into consideration, it is important that we investigate the classical simulation complexities of the various subsets $\mathcal C_\nu$ under notions of approximate classical simulation. For example, consider the cases labeled in Table \ref{mainTable} by $\PH$. Under the plausible complexity assumption that the polynomial hierarchy does not collapse, are these extended Clifford circuits still hard to sample from if we required only that the circuits are weakly simulable under small variation distance? We leave open this question, though we note that some progress has been made for the case where the same question is asked of other restricted models of quantum computation. For example, Aaronson and Arkhipov have showed that for boson sampling, hardness of weak simulation under small variation distance holds if we assume that certain unproven conjectures, like the \textit{Permanent-of-Gaussians Conjecture}, and the \textit{Permanent Anti-Concentration Conjecture}, are true \cite{AaronsonArkhipov}. Similar results have been obtained for IQP circuits \cite{bremner2015average,bremner2016achieving, fujii2016noise}, and it will be interesting to prove similar hardness results for extended Clifford circuits.

\section*{Acknowledgements}
I thank Anand Natarajan for discussions on the proofs of Theorems \ref{Thm1} and \ref{Thm3}, and Adam Bouland for discussions on the proofs of Theorems \ref{Thm1} and \ref{Thm2}. I also thank Scott Aaronson for useful insights and for teaching the seminar class 6.S899 on Physics and Computation at MIT, which led to this research. The author is supported by the National Science Scholarship from the Agency for Science, Technology and Research (A*STAR).

\bibliographystyle{ieeetr}
\bibliography{Bib}

\begin{thebibliography}{10}

\bibitem{Gottesman98}
D.~Gottesman, ``The {H}eisenberg representation of quantum computers,'' {\em
  Talk at International Conference on Group Theoretic Methods in Physics},
  arXiv: quant-ph/9807006 1998.

\bibitem{GottesmanThesis}
D.~Gottesman, {\em Stabilizer Codes and Quantum Error Correction}.
\newblock PhD thesis, California Institute of Technology, 1997.

\bibitem{Raussendorf}
R.~Raussendorf, D.~E. Browne, and H.~J. Briegel, ``Measurement-based quantum
  computation on cluster states,'' {\em Phys. Rev. A}, vol.~68, p.~022312, Aug
  2003.

\bibitem{RaussendorfBriegel}
R.~Raussendorf and H.~J. Briegel, ``A one-way quantum computer,'' {\em Phys.
  Rev. Lett.}, vol.~86, pp.~5188--5191, May 2001.

\bibitem{Pusey}
M.~Pusey, ``Stabilizer notation for {S}pekkens' toy theory,'' {\em Found Phys},
  vol.~42, no.~688-708, 2012.

\bibitem{Spekkens}
R.~W. Spekkens, ``Evidence for the epistemic view of quantum states: A toy
  theory,'' {\em Phys. Rev. A}, vol.~75, p.~032110, Mar 2007.

\bibitem{penney2016quantum}
M.~D. Penney, D.~E. Koh, and R.~W. Spekkens, ``Quantum circuit dynamics via
  path integrals: Is there a classical action for discrete-time paths?,'' {\em
  arXiv preprint arXiv:1604.07452}, 2016.

\bibitem{Jozsa}
R.~Jozsa and M.~Van~den Nest, ``Classical simulation complexity of extended
  {C}lifford circuits,'' {\em Quantum Info. Comput.}, vol.~14, pp.~633--648,
  2014.

\bibitem{preskill2012quantum}
J.~Preskill, ``Quantum computing and the entanglement frontier,'' {\em arXiv
  preprint arXiv:1203.5813}, 2012.

\bibitem{aaronsonchen}
S.~Aaronson and L.~Chen, ``Complexity-theoretic foundations of quantum
  supremacy experiments,'' {\em arXiv preprint arXiv:1612.05903}, 2016.

\bibitem{AaronsonArkhipov}
S.~Aaronson and A.~Arkhipov, ``The computational complexity of linear optics,''
  in {\em Proceedings of the Forty-third Annual ACM Symposium on Theory of
  Computing}, STOC '11, (New York, NY, USA), pp.~333--342, ACM, 2011.

\bibitem{bremner2010classical}
M.~J. Bremner, R.~Jozsa, and D.~J. Shepherd, ``Classical simulation of
  commuting quantum computations implies collapse of the polynomial
  hierarchy,'' in {\em Proceedings of the Royal Society of London A:
  Mathematical, Physical and Engineering Sciences}, p.~rspa20100301, The Royal
  Society, 2010.

\bibitem{bremner2015average}
M.~J. Bremner, A.~Montanaro, and D.~J. Shepherd, ``Average-case complexity
  versus approximate simulation of commuting quantum computations,'' {\em arXiv
  preprint arXiv:1504.07999}, 2015.

\bibitem{bremner2016achieving}
M.~J. Bremner, A.~Montanaro, and D.~J. Shepherd, ``Achieving quantum supremacy
  with sparse and noisy commuting quantum computations,'' {\em arXiv preprint
  arXiv:1610.01808}, 2016.

\bibitem{farhi2016quantum}
E.~Farhi and A.~W. Harrow, ``Quantum supremacy through the quantum approximate
  optimization algorithm,'' {\em arXiv preprint arXiv:1602.07674}, 2016.

\bibitem{morimae2014hardness}
T.~Morimae, K.~Fujii, and J.~F. Fitzsimons, ``Hardness of classically
  simulating the one-clean-qubit model,'' {\em Physical review letters},
  vol.~112, no.~13, p.~130502, 2014.

\bibitem{mehraban2015computational}
S.~Mehraban, ``Computational complexity of some quantum theories in $1+ 1$
  dimensions,'' {\em arXiv preprint arXiv:1512.09243}, 2015.

\bibitem{aaronson2016computational}
S.~Aaronson, A.~Bouland, G.~Kuperberg, and S.~Mehraban, ``The computational
  complexity of ball permutations,'' {\em arXiv preprint arXiv:1610.06646},
  2016.

\bibitem{NielsenChuang}
M.~Nielsen and I.~Chuang, {\em Quantum Computation and Quantum Information}.
\newblock Cambridge University Press, 2000.

\bibitem{Terhal}
B.~Terhal and D.~DiVincenzo, ``Adaptive quantum computation, constant depth
  quantum circuits and {A}rthur-{M}erlin games,'' {\em Quant. Inf. Comp.},
  vol.~4, no.~2, p.~134, 2004.

\bibitem{AaronsonGottesman}
S.~Aaronson and D.~Gottesman, ``Improved simulation of stabilizer circuits,''
  {\em Phys. Rev. A}, vol.~70, p.~052328, Nov 2004.

\bibitem{Eisert}
A.~Mari and J.~Eisert, ``Positive {W}igner functions render classical
  simulation of quantum computation efficient,'' {\em Phys. Rev. Lett.},
  vol.~109, p.~230503, Dec 2012.

\bibitem{Veitch}
V.~Veitch, N.~Wiebe, C.~Ferrie, and J.~Emerson, ``Efficient simulation scheme
  for a class of quantum optics experiments with non-negative {W}igner
  representation,'' {\em New Journal of Physics}, vol.~15, no.~1, p.~013037,
  2013.

\bibitem{Ni}
X.~Ni, O.~Buerschaper, and M.~Van~den Nest, ``A non-commuting stabilizer
  formalism,'' {\em Journal of Mathematical Physics}, vol.~56, no.~5, 2015.

\bibitem{fujii2016noise}
K.~Fujii, ``Noise threshold of quantum supremacy,'' {\em arXiv preprint
  arXiv:1610.03632}, 2016.

\bibitem{Sipser}
M.~Sipser, {\em Introduction to the Theory of Computation, 3rd Ed.}
\newblock Course Technology, 2012.

\bibitem{Bremner}
M.~J. Bremner, R.~Jozsa, and D.~J. Shepherd, ``Classical simulation of
  commuting quantum computations implies collapse of the polynomial
  hierarchy,'' {\em Proceedings of the Royal Society of London A: Mathematical,
  Physical and Engineering Sciences}, vol.~467, no.~2126, pp.~459--472, 2010.

\bibitem{Aaronson}
S.~Aaronson, ``Quantum computing, postselection, and probabilistic
  polynomial-time,'' {\em Proceedings of the Royal Society of London A:
  Mathematical, Physical and Engineering Sciences}, vol.~461, no.~2063,
  pp.~3473--3482, 2005.

\bibitem{Toda}
S.~Toda, ``{PP} is as hard as the polynomial-time hierarchy,'' {\em SIAM J.
  Comput.,}, vol.~20, no.~5, pp.~865--877, 1991.

\end{thebibliography}

\appendix
\newpage

\section{Rules for proving results in Table \ref{mainTable}}
\label{sec:SimplifyingRules}

In this section, we show that the entries in Table \ref{mainTable} that contain boxed symbols (for example, $\boxed\PH$) can be used to deduce all the other entries in the table. Therefore, for a complete proof of the results in the table, it will suffice to prove just Theorems \ref{Thm1}--\ref{Thm6} as well as JV 1--6 (save JV1 and JV6). This is a straightforward consequence of a couple of simple rules (cf \cite{Jozsa}), which we state explicitly here:
\begin{itemize}
\item If the classical simulation of a set of computational tasks $\mathcal A$ is efficient, then the classical simulation of any subset of $\mathcal A$ would also be efficient.
\item If the classical simulation of a set of computational tasks $\mathcal A$ is hard ($\#\P$-hard, $\QC$-hard or $\PH$-collapsing in the sense described above), then the classical simulation of any superset of $\mathcal A$ would also be similarly hard.
\item The set of computational tasks with IN(BITS) is a subset of the same set of tasks with IN(PROD). Write this as IN(BITS) $\subset$ IN(PROD). Similarly, OUT(BITS) $\subset$ OUT(PROD), NONADAPT $\subset$ ADAPT.
\item If the strong simulation of a set of tasks is efficient, then so are the $\strong(1)$, $\strong(n)$ and $\weak(n)$ simulations of that set. If any of the latter three notions of simulation is efficient, then $\weak(1)$-simulation of that set is also efficient (as illustrated in Figure \ref{fig:notionsOfxSimulation}). In the opposite direction, if $\strong(1)$ or $\strong(n)$ simulation is $\#\P$-hard, then so is strong simulation. Similarly, if $\weak(1)$ simulation is $\QC$-hard, then $\weak(n)$-simulation is also $\QC$-hard. Note that $\#\P$-hardness holds only for strong notions of simulation, and $\QC$-hardness holds only for weak notions of simulation.
\end{itemize}

\section{Proof of Theorem 1}
\label{sec:thm1}

A 3-CNF formula $f$ (i.e.\@ a Boolean formula in conjunctive normal form \cite{Sipser}) with $n$ variables and $N$ clauses is of the form
\begin{equation}
\label{3cnf_formula}
f(x_1,\ldots, x_n) = (a_{11} \vee a_{12} \vee  a_{13}) \wedge (a_{21} \vee a_{22} \vee  a_{23}) \wedge \ldots \wedge (a_{N1} \vee a_{N2} \vee  a_{N3}),
\end{equation} 
where each $a_{ij} \in \{x_1 ,\ldots, x_n, \bar x_1, \ldots, \bar x_n\}$. We shall assume that every variable $x_1,\ldots, x_n$ appears in the formula for $f$, so that $n\leq 3N$, i.e.\@ $n=O(N)$.

We define $\AbsSAT$ to be the following problem: Given a 3-CNF formula $f:\uint^n \rar \uint$, compute
$$ S(f)= \left| \sum_{x\in\uint^n} (-1)^{f(x)}  \right|.$$ 

We shall denote $\#_i(f) = \left| \{ x|f(x)=i \}\right|$ for $i=0,1$. Then $S(f) = |\#_0(f)-\#_1(f)|$.

\begin{lemma}
$\AbsSAT$ is $\#\P$-hard.
\end{lemma}
\begin{proof}
We shall construct a reduction from the $\#\P$-complete problem $\#\SAT$ to $\AbsSAT$. Given a  $\#\SAT$-instance $\phi(x_1,\ldots, x_n)$, introduce a new variable $y$ and define the Boolean formula $$\tilde\phi(x_1,\ldots, x_n,y)=\phi(x_1,\ldots, x_n)\vee y.$$ Let $A(\varphi)$ denote the set of satisfying assignments to a Boolean formula $\varphi$. Then $$A(\tilde \phi) = \{(x_1,\ldots,x_n,0)|(x_1,\ldots,x_n)\in A(\phi)\} \cup \{(x_1,\ldots,x_n,1)|(x_1,\ldots,x_n)\in\uint^n\}.$$ 
Hence, $\#_1(\tilde\phi)=\#_1(\phi)+2^n$, and $\#_0(\tilde\phi)=2^{n+1}-\#_1(\tilde \phi) = 2^n- \#_1(\phi)$.
This gives
$$S(\tilde \phi) = |\#_0(\tilde\phi)-\#_1(\tilde\phi)| =|2^n -\#_1(\phi) - \#_1(\phi)-2^n | = 2\#_1(\phi).$$
Solving the $\AbsSAT$ instance $\tilde\phi(x_1,\ldots, x_n,y)$ gives $S(\tilde \phi)$, from which $\#_1(\phi)$ can be found. Therefore, $\AbsSAT$ is $\#\P$-hard.

\end{proof}

\begin{theorem}
\label{Thm1}
Let $\nu =$ (IN(PROD), NONADAPT, OUT(BITS)). Then the $\strong(n)$-simulation of $\mathcal C_\nu$ is $\#\P$-hard.
\end{theorem}

\begin{proof}

Assume that there exists an efficient $\strong(n)$-simulation $S$ of $\mathcal C_\nu$. We'll use $S$ to construct an efficient algorithm for $\AbsSAT$: On input $f:\uint^n\rar\uint$, given as a 3-CNF formula with $N$ clauses, where $n=O(N)$, construct a quantum circuit $Q_f$, consisting of only the basic Clifford gates and $T$ gates, that acts on the following computational basis states as follows: (See Lemma \ref{3CNFCliffordT} for the details of such a construction)
$$Q_f \ket{x_1,\ldots, x_n, 0} | \vec 0 \rangle_A = \ket{ x_1,\ldots, x_n, f(x_1,\ldots,x_n)} | \vec 0 \rangle_A.$$

Let $K$ be the number of $T$ gates in $Q_f$. For the $j$th $T$ gate (acting on the $l_j$th line), for $j=1,\ldots,K$, introduce an ancilla line $a_j$, and replace the $T$ gate with the CNOT gate $CX_{l_j,a_j}$. Call the resulting circuit $A_f$. It is straightforward to check that if each ancilla wire is initialized to the state $\ket{\pi/4}$, and measured at the end of the computation, and if the measurement outcomes are $0\ldots 0$, then the non-ancilla registers of $A_f$ would implement $Q_{ f}$.
Hence, ignoring the ancilla registers, for the above measurement outcomes, we have
$A_f : |x_1,\ldots, x_n,y\rangle|\vec 0\rangle_A  \mapsto |x_1,\ldots, x_n,y\oplus f(x_1,\ldots, x_n)\rangle|\vec 0\rangle_A$.

Let $M_{f}$ be the following circuit: \def\Af{A_f} 

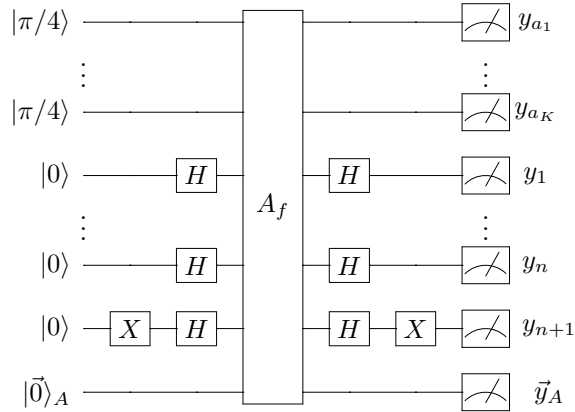
\begin{align}
\Qcircuit @C=1em @R=1em {
& & & &\lstick{\ket{\pi/4}}      & \qw      & \qw & \multigate{7}{\Af} & \qw & \qw & \meter & y_{a_1} \\
& & & &                    \vdots              &    &      &  &                 &  & \vdots &  \\
& & & &\lstick{\ket{\pi/4}}      & \qw      & \qw & \ghost{\Af} & \qw & \qw & \meter & y_{a_K}\\
& & & &\lstick{\ket{0}}      & \qw      & \gate H & \ghost{\Af} & \gate H & \qw & \meter & y_{1} \\
& & & &                   \vdots               &     &       & \ &                &  &\vdots &  \\
& & & &\lstick{\ket{0}}      & \qw      & \gate H & \ghost{\Af} & \gate H & \qw & \meter & y_{n}\\
& & & &\lstick{\ket{0}}      & \gate X      & \gate H & \ghost{\Af} & \gate H & \gate X & \meter &  \quad y_{n+1}\\
& & & &\lstick{|\vec 0\rangle_A}      &  \qw     & \qw & \ghost{\Af} &\qw  & \qw & \meter &  \quad \vec y_{A}
}
\nonumber
\end{align}

If we postselect on the outcomes  $y_{a_1}\ldots y_{a_K}= 0\ldots 0$ for the ancilla registers, the nonancilla registers evolve as follows:

\begin{eqnarray*}
|0\ldots 0 ,0\rangle \anc 0 &\rar & |0\ldots 0 ,1\rangle \anc 0\\
&\rar &   \f 1{\sqrt{2^{n+1}}} \sum_x \ket x (\ket 0 -\ket 1) \anc 0\\
&\rar &   \f 1{\sqrt{2^{n+1}}} \sum_x \ket x (\ket {f(x)} -\ket {1\oplus f(x)}) \anc 0\\
&&=  \f 1{\sqrt{2^{n+1}}} \sum_x  (-1)^{f(x)}\ket x (\ket 0 -\ket 1) \anc 0\\
&\rar &  \f 1{2^n} \sum_{xy}  (-1)^{f(x)+x\cdot y}\ket y \ket 1\anc 0 \\
&\rar &  \f 1{2^n} \sum_{xy}  (-1)^{f(x)+x\cdot y}\ket y \ket 0\anc 0.
\end{eqnarray*}

Hence, the conditional probability of obtaining the all-zero string given that the ancilla measurements also reveal the all-zero string is
$$ \Pr(0_1 \ldots 0_{n+1}, \vec 0_A| 0_{a_1} \ldots 0_{a_K}) = \left| \f 1{2^n} \sum_x (-1)^{f(x)} \right|^2. $$

But the LHS of the above expression is equal to
$$ \Pr(0_1 \ldots 0_{n+1}, \vec 0_A | 0_{a_1} \ldots 0_{a_K}) =\f{ \Pr(0_{a_1} \ldots 0_{a_K}, 0_1 \ldots 0_{n+1}, \vec 0_A) }{\Pr(0_{a_1} \ldots 0_{a_K}) }.$$
Now, $\Pr(0_{a_1} \ldots 0_{a_K})= 1/2^K$, since each ancilla bit has a probability of $1/2$ of being measured zero. 

Simplifying the above expressions, we get
$$ \left| \sum_x (-1)^{f(x)} \right| = 2^{n+K/2} \sqrt{\Pr(0_{a_1} \ldots 0_{a_K}, 0_1 \ldots 0_{n+1}, \vec 0_A)}.$$ 

But $\Pr(0_{a_1} \ldots 0_{a_K},0_1 \ldots 0_{n+1}, \vec 0_A)$ is a joint outcome probability, and hence can be obtained by running $S$ on $\langle M_{f}, 00\ldots 0\rangle$. (The input to $S$ is valid since $M_{ f}$ is a nonadaptive Clifford circuit with product state inputs.) Hence, the procedure given is an efficient algorithm for $\AbsSAT$. Since $\AbsSAT$ is $\#\P$-hard, this implies that $\mathcal C_\nu$ is $\#$P-hard as well.
\end{proof}

\section{Proof of Theorem 2}
\label{sec:thm2}

\begin{theorem}
\label{Thm2}
Let $\nu =$ (IN(BITS), ADAPT, OUT(BITS)). Then the $\strong(n)$-simulation of $\mathcal C_\nu$ is $\#\P$-hard.
\end{theorem}
\begin{proof}
Assume that there exists an efficient $\strong(n)$-simulation $S$ of $\mathcal C_\nu$. We'll use $S$ to construct an efficient algorithm $M$ for $\# \SAT$, i.e.\@ given as input a 3-CNF formula $f:\{0,1\}^n \rightarrow \{0,1\}$, our goal is to find $\# f = \sum_x f(x)$.

M = ``On input $f:\{0,1\}^n \rightarrow \{0,1\}$, given as a 3-CNF formula,
\begin{enumerate}
\item[1.] Construct a classical circuit $C_f$ consisting of only Toffoli gates that acts on the following computational basis states as follows: (see Lemma \ref{3CNFToffoli} for the details of this construction)
$$C_f (x_1,\ldots, x_n, 1, \vec 1_A) = (x_1,\ldots, x_n, f(x_1,\ldots,x_n), \vec 1_A).$$
\item[2.] Simulate $C_{f}$ with a Clifford circuit from $\mathcal A$: replace each Toffoli gate $T_{abc}(x,y,z) = (x,y,z\oplus xy)$ acting on lines $a,b,c$ with $(CX_{bc})^x M_a(x)$. Call the resulting quantum circuit $A_f$. The circuit $A_f$ acts on computational basis states as follows:
$$ A_f \ket{x_1,\ldots,x_n,1}|\vec 1 \rangle_A \rightarrow \ket{x_1,\ldots,x_n,f(x_1,\ldots,x_n)} |\vec 1 \rangle_A. $$
By applying $X$ gates (expressed as $X = HS^2H$) to the appropriate lines at the input and output of $A_f$, let $A'_f$ be the circuit that acts on computational basis states as follows:
$$ A'_f \ket{x_1,\ldots,x_n,0}|\vec 0 \rangle_A \rightarrow \ket{x_1,\ldots,x_n,f(x_1,\ldots,x_n)} |\vec 0 \rangle_A. $$
\item[3.] Let $G_f$ be the following circuit:
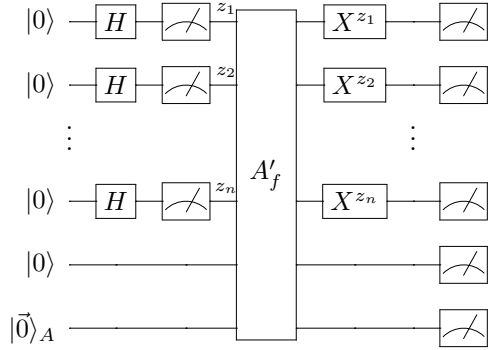
\begin{align}
\Qcircuit @C=1em @R=1em {
& &\lstick{\ket{0}} & \gate H & \meter & \multigate{6}{A'_f}_{z_1} & \gate {X^{z_1}} & \qw & \meter \\ 
& &\lstick{\ket{0}} & \gate H & \meter & \ghost{A'_f}_{z_2} & \gate {X^{z_2}} & \qw  & \meter\\
& &\vdots & & &  &  & \vdots \\
& & & & &  &  &  \\
& &\lstick{\ket{0}} & \gate H & \meter & \ghost{A'_f}_{z_n} & \gate {X^{z_n}} & \qw  & \meter\\
& &\lstick{\ket{0}} & \qw & \qw &\ghost{A'_f} & \qw & \qw  & \meter \\
& &\lstick{\anc 0} & \qw & \qw &\ghost{A'_f} & \qw & \qw  & \meter
\nonumber
}
\end{align}

\item[4.] Feed $\langle G_f, 00\ldots 01 \vec 0_A\rangle$ into $S$ to find $p=p(00\ldots 01\vec 0_A)$, the probability that the output is $00\ldots 01\vec 0_A$.
\item[5.] Output $\# f = 2^n p$."
\end{enumerate}

A straightforward calculation shows that the output of $G_f$ on input $|00\ldots 0\rangle \anc 0$  is $|0,\ldots, 0, f(z)\rangle \anc 0$ if the intermediate measurement results are $z= z_1\ldots z_n$
Hence, 
\begin{eqnarray}
	p = p(0\ldots 0,1,\vec 0_A) &=& \sum_z p(0\ldots 0,1,\vec 0_A|z_1 \ldots z_n) p(z_1 \ldots z_n) \nonumber \\
	&=& \sum_z |\langle 0\ldots 0 1,\vec 0_A|0\ldots 0, f(z), \vec 0_A \rangle|^2 \frac{1}{2^n} \nonumber\\
	&=& \frac{1}{2^n} \sum_x f(x).
\end{eqnarray}
Hence, the output of $M$ is $ 2^n p = \# f$.

\end{proof}

\section{Proof of Theorem 3}
\label{sec:thm3}

We follow a proof similar to that given in \cite{Bremner} that shows that if IQP circuits can be efficiently classically simuated in the weak sense, then the polynomial hierarchy collapses. Recall that the $T$ gate is given by $T = \diag(1,e^{i\pi/4})$. We first consider the following gadget $\mathcal G$:
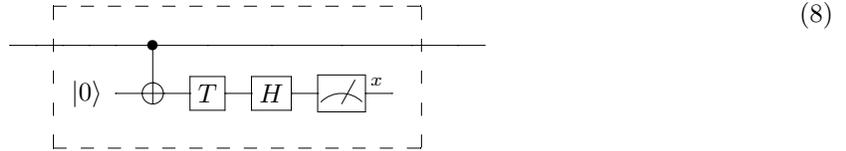
\begin{align} \label{eq:gadgetG}
\Qcircuit @C=1em @R=1em {
& & & & & & & & & & & \\
& \qw & \qw & \qw & \qw & \ctrl{1} &  \qw & \qw & \qw & \qw & \qw &\qw \gategroup{1}{3}{4}{10}{.7em}{--} \\
& & & &\lstick{\ket{0}} & \targ & \gate{T} & \gate{H} & \meter & \ghost{1}_x & &   \\
& & & & & & & & & & &
}
\end{align}

\begin{lemma}
\label{Tgadget}
$$\mathcal G : \ket \psi \mapsto \begin{cases} T\ket \psi & \textrm{if } x=0 \\ ZT\ket \psi & \textrm{if } x=1 \end{cases}$$
\end{lemma}
\begin{proof}
Applying the unitary gates in the circuit to the state $\ket \psi \ket 0$ gives $\frac 1{\sqrt 2} \left[ (T\ket \psi)\ket 0 + (ZT\ket \psi)\ket 1 \right]$. Hence, we get the desired states when the ancilla wire is measured.
\end{proof}

From the proof of Lemma \ref{Tgadget}, we note that the measurement outcomes $x=0,1$ occur with an equal probability. Note that if $x=0$, then $\mathcal G$ would have implemented the $T$ gate.

\begin{lemma}
\label{lemma:simulateQ}
Let $Q$ be an arbitrary quantum circuit comprising the basic Clifford gates and $T$ gates. Let $\nu = $ (IN(BITS), NONADAPT, OUT(PROD)). Then $Q$ with postselection can be weakly simulated by $\mathcal C_\nu$ with postselection.
\end{lemma}

\begin{proof}
We first show how we can simulate the circuit $Q$ using circuits from $\mathcal C_\nu$ with postselection. For each $T$ gate in $Q$, we replace it by the gadget $\mathcal G$ defined above. If the number of $T$ gates is $s$, then this procedure produces a new circuit $C$ with $s$ new lines. Now, note that the new circuit $C$ belongs to the class $C_\nu$ since the $HT$ gates together with the computational basis measurements implement a product measurement. Now, if we postselect on outcome 0 for all the measurements in the new lines, then each gadget $\mathcal G$ would implement the $T$ gate. Hence,  $\mathcal C_\nu$ with postselection would weakly simulate $Q$. Now, since we have the resource of postselection, it follows that $Q$ with postselection can be weakly simulated by $\mathcal C_\nu$ with postselection.
\end{proof}

We now make the following definition (recall notation in Eq.\,\eqref{marginals}: we use similar notation for conditional probabilities) to capture the power of subsets of Clifford computational tasks with postselection. 

\begin{definition} ($\postcnu$)
Let $\mathcal C_\nu$ be a subset of Clifford computational tasks. A language $L \in \postcnu$ if there exists an error tolerance $0<\epsilon< \frac 12$, and a uniform family $\{C_w\}_w$ of circuits in $\mathcal C_\nu$ with $n+p(n)$ lines (call these lines $l_1,\ldots, l_n, a_1,\ldots, a_p$, where $p=p(n)$), where $n=|w|$  and $p$ is some polynomial, such that 
$$ p_{C_w}^{\{a_1,\ldots,a_N\}} (00\ldots 0) > 0, $$
$$ w \in L \implies p_{C_w}^{\{l_1\}|\{a_1,\ldots,a_N\}} (1|00\ldots 0) \geq 1-\epsilon,  $$
\begin{equation}
\label{postcnuCw}
w \notin L \implies p_{C_w}^{\{l_1\}|\{a_1,\ldots,a_N\}} (0|00\ldots 0) \geq 1-\epsilon.
\end{equation}
\end{definition}

We will use the definition of $\postBQP$ given in \cite{Bremner}, which allows for multiple postselected lines. Note that this is equivalent to the definition given in \cite{Aaronson} where $\postBQP$ was introduced, which allows for only single lines. We now show that the class just defined is equal to $\postBQP$.
\begin{lemma}
\label{postcnupostBQP}
Let $\nu = $ (IN(BITS), NONADAPT, OUT(PROD)). Then, $\postcnu = \postBQP$.
\end{lemma}

\begin{proof}
The forward direction is immediate, since extended Clifford circuits are a special case of general quantum circuits. To prove the backward direction, let $L\in \postBQP$. Then there exists an error tolerance $0<\epsilon< \frac 12$, and a uniform family $\{Q_w\}_w$ of quantum circuits consisting of the basic Clifford gates and $T$ gates with $n+p(n)$ lines (call these lines $l_1,\ldots, l_n, b_1,\ldots, b_N$, where $p=p(n)$), where $n=|w|$  and $p$ is some polynomial, such that 
$$ p_{Q_w}^{\{b_1,\ldots,b_N\}} (00\ldots 0) > 0, $$
$$ w \in L \implies p_{Q_w}^{\{l_1\}|\{b_1,\ldots,b_p\}} (1|00\ldots 0) \geq 1-\epsilon,  $$
$$ w \notin L \implies p_{Q_w}^{\{l_1\}|\{b_1,\ldots,b_N\}} (0|00\ldots 0) \geq 1-\epsilon.  $$
By Lemma \ref{lemma:simulateQ}, for each $Q_w$, there exists an extended Clifford circuit $C_w\in \mathcal C_\nu$ that, with postselection, simulates $Q_w$ with postselection. If $s$ is the number of $T$ gates in $Q_w$, then $C_w$ has $n+p(n)+s$ lines. Postselecting on the last $p(n)+s$ lines, it follows that the set of circuits $\{C_w\}$ satisfies the definition given for $\postcnu$. Hence, $L\in \postcnu$.
\end{proof}

\begin{lemma}
\label{postcnupostBPP}
 Let $\nu = $ (IN(BITS), NONADAPT, OUT(PROD)). If $\mathcal C_\nu \in \P\weak(n)$, then $\postcnu \subseteq \postBPP$.
\end{lemma}

\begin{proof}
Let $L \in \postcnu$. Then there exists an error tolerance $0<\epsilon< \frac 12$, and a uniform family $\{C_w\}_w$ of circuits in $\mathcal C_\nu$ with $n+p(n)$ lines (call these lines $l_1,\ldots, l_n, a_1,\ldots, a_p$, where $p=p(n)$), where $n=|w|$  and $p$ is some polynomial, such that Eq.\,\eqref{postcnuCw} holds.

But $\mathcal C_\nu \in \P\weak(n)$. Hence, for all circuits $Q_w \in \mathcal C_\nu$, there exists a classical randomized circuit $C_w$ with $n+p$ lines such that 
$$ p_{Q_w}^{\{l_1,\ldots, l_n, a_1,\ldots, a_p\}} (y) = p_{C_w}^{\{l_1,\ldots, l_n, a_1,\ldots, a_p\}} (y).$$

For any subsets $I, J\subseteq [n]$ of lines, similar relations hold for marginal probabilities and conditional probabilities: $ p_{Q_w}^I (y) = p_{C_w}^I (y)$ and $p_{Q_w}^{I|J} (y|z) = p_{C_w}^{I|J} (y|z)$. 
This implies that
$$ p_{Q_w}^{\{l_1\}|\{a_1,\ldots,a_p\}} (1|00\ldots 0) = p_{C_w}^{\{l_1\}|\{a_1,\ldots,a_p\}} (1|00\ldots 0),$$
and hence $Q_w$ obey Eq.\,\eqref{postcnuCw}. This implies that $L \in \postBPP$. Therefore, $\postcnu \subseteq \postBPP$.
\end{proof}

\begin{theorem}
\label{Thm3}
Let $\nu = $ (IN(BITS), NONADAPT, OUT(PROD)). If $\mathcal C_\nu \in \P\weak(n)$, then $\PH$ collapses to the third level.
\end{theorem}

\begin{proof}
By Lemmas \ref{postcnupostBQP} and \ref{postcnupostBPP}, if $\mathcal C_\nu \in \P\weak(n)$, then $\postBPP \supseteq \postcnu = \postBQP$. By Aaronson's Theorem, $\PP = \postBQP$ \cite{Aaronson}, and by Toda's Theorem, $\PH \subseteq \P^{\#\P}$ \cite{Toda}. Hence, we get the following string of inclusions:
$$ \PH \subseteq \P^{\#\P} = \P^\PP =\P^\postBQP = \P^{\postcnu} \subseteq \P^\postBPP \subseteq \P^{\BPP^\NP} = \BPP^\NP \subseteq \Sigma_3^p ,$$
which implies that $\PH$ collapses to the third level.
\end{proof}

\section{Proof of Theorem 4}
\label{sec:thm4}
 
 Consider the proof of Theorem \ref{Thm1}. Note that the circuit $M_{f}$ is unitary. Hence, an even stronger result than 
Theorem \ref{Thm1} is true: if we replaced nonadaptive circuits with unitary ones (call this UNITARY), the simulation complexity is still $\#\P$-hard. In other words, 
\begin{lemma}
Let $\nu =$ (IN(PROD), UNITARY, OUT(BITS)). Then the $\strong(n)$-simulation of $\mathcal C_\nu$ is $\#\P$-hard.
\end{lemma}

The $\strong(n)$-simulation of $\mathcal C_\nu$ is equivalent to the following problem:
\newline\textbf{Input}: $\langle T, y\rangle$, where $T = ( \ket{x}, B, \ket{\alpha})$, $B$ is a unitary circuit, $x,y \in \{0,1\}^n$, $\alpha = \alpha_1\ldots\alpha_n$ and each $\ket {\alpha_i} \in \mathbb C_2$.\newline
\textbf{Output}: $p_T(y) = |\langle \alpha_1^{y_1}\ldots\alpha_n^{y_n}| B|x\rangle |^2$.

Now, let $\mu =$ (IN(BITS), UNITARY, OUT(PROD)), then the $\strong(n)$-simulation of $\mathcal C_\mu$ is equivalent to the following problem: 
\newline
\textbf{Input}: $\langle T', y\rangle$, where $T' = ( \ket{\alpha_1^{y_1}\ldots\alpha_n^{y_n}}, B^\dag, \{ I_2 \}_i)$ , $B$ is a unitary circuit, $y \in \{0,1\}^n$ and $I_2$ is the $2 \times 2$ identity gate.
\newline
\textbf{Output}: $p_{T'}(x) = |\langle x| B^\dag|\alpha_1^{y_1}\ldots\alpha_n^{y_n}\rangle |^2 =|\langle \alpha_1^{y_1}\ldots\alpha_n^{y_n}| B|x\rangle |^2= p_T(y)$.

Since both problem instances can be transformed easily to each other, and since both problems involve calculating the same quantity, we conclude that the $\strong(n)$-simulation of $\mathcal C_\mu$ is also $\#\P$-hard. If it is $\#\P$-hard to simulate this class of unitary circuits, then it must be $\#\P$-hard to simulate the same class but with unitary circuits replaced by nonadaptive circuits. Therefore, we obtain the following theorem:

\begin{theorem}
\label{Thm4}
Let $\nu =$ (IN(BITS), NONADAPT, OUT(PROD)). Then the $\strong(n)$-simulation of $\mathcal C_\nu$ is $\#\P$-hard.
\end{theorem}

\section{Proof of Theorem 5}
\label{sec:thm5}

\begin{theorem}
\label{Thm5}
Let $\nu =$ (IN(PROD), NONADAPT, OUT(PROD)). Then $\mathcal C_\nu \in \P\strong(1)$.
\end{theorem}
\begin{proof}

We use the following notation: for any single-qubit operator $O$, let $O_1 = O \otimes I \otimes \ldots \otimes I$. 
Given a Clifford computational task $T = (|\alpha_1\ldots \alpha_n\rangle, B, \{U, I,\ldots, I\}) \in \mathcal C_\nu$, and a bit $i\in \{0,1\}$, we shall describe an algorithm to compute $p_i :=p_T^{\{1\}}(i)$. WLOG, $B$ is a unitary circuit. 

Since $p_0+p_1=1$, it suffices to be able to calculate $p_0-p_1$ efficiently. 
By Born's rule, this is given by
\begin{equation}
\label{diff}
p_0 - p_1 = \langle \alpha | B^\dag (UZU^\dag)_1 B\ket\alpha.
\end{equation}

Since the Pauli matrices $\{\sigma^i\}_i$ form a basis for the set of $2\times 2$ matrices , we can write 
$$ U = \sum_{i=0}^3 a_i \sigma^i ,$$
for some $a_i \in \mathbb C$.
Hence, 
$$ UZU^\dag = \sum_{i,j=0}^3 a_i \bar{a}_j \sigma^i Z \sigma^j.$$

But $\sigma^i Z \sigma^j$ is a Pauli operator. Since the basic Clifford gates map Pauli operators to Pauli operators, 
$$ B^\dag (\sigma^i Z\sigma^j)_1 B = \gamma_{ij} P_1^{ij}\otimes \ldots \otimes P_n^{ij}.$$

Putting this into Eq.\,\eqref{diff}, we get an expression for $p_0-p_1$.
\begin{eqnarray}
\label{expressionFordiff}
p_0-p_1 &=& \sum_{i,j=0}^3 a_i \bar a_j \gamma_{ij} \bra{\alpha_1\ldots\alpha_n} P_1^{ij}\otimes \ldots \otimes P_n^{ij}  \ket{\alpha_1\ldots\alpha_n} \nonumber\\
&=& \sum_{i,j=0}^3 a_i \bar a_j \gamma_{ij} \prod_{k=1}^n \bra{\alpha_k} P_k^{ij} \ket{\alpha_k}.
\end{eqnarray}

We now analyze the running time of our algorithm. Computing $\gamma_{ij} P_1^{ij}\otimes \ldots \otimes P_n^{ij}$ takes $O(n^2)$-time. The formula given in Eq.\,\eqref{expressionFordiff} involves a sum of 9 terms. Each term involves computing $n$ expectation values of $2\times 2$ matrices. Hence, this step takes $O(n)$-time. Overall, the algorithm runs in $O(n^2) = O(N^2)$-time, where $N$ is the number of gates in the circuit (which we assumed to contain no extraneous lines). Hence, $\mathcal C_\nu \in \P\strong(1)$.

\end{proof}

\section{Proof of Theorem 6}
\label{sec:thm6}

\begin{theorem}
\label{Thm6}
Let $\nu =$ (IN(BITS), ADAPT, OUT(PROD)). Then $\mathcal C_\nu \in \P\weak(1)$.
\end{theorem}
\begin{proof}
This is a special case of the results in Section VIIC of \cite{AaronsonGottesman}, which showed that an IN(BITS), NONADAPT, OUT(BITS) circuit containing $d$ non-Clifford gates, where each gate acts on at most $b$ qubits, can be classically simulated in the WEAK(1) sense in $O(4^{2bd} n+n^2)$-time. In our case, the circuits in $ C_\nu$ can be thought of as containing exactly one non-Clifford gate on the first wire just before the computational-basis measurement. Hence, $d=b=1$, which implies that the algorithm runs in $O(n^2)$-time. This concludes the proof that $\mathcal C_\nu \in \P\weak(1)$.

\end{proof}

\section{Constructing circuits for 3-CNF formulas}

In the proof of Theorem 1, we used the fact that given a 3-CNF formula $f:\{0,1\}^n \rar \{0,1\}$, we can efficiently construct a quantum circuit $A_f$ comprising only the basic Clifford operations and $T$ gates, which acts on the following computational basis states as follows:
\begin{equation}
 A_f \ket x \ket 0 \ket 0_A =  \ket x \ket { f(x)} \ket 0_A,
\end{equation}
where $x\in\{0,1\}^n$ and $\ket{\cdot}_A$ is an ancilla register of size $O(n)$. 

A similar fact was used in the proof of Theorem 2, namely that given a 3-CNF formula $f:\{0,1\}^n \rar \{0,1\}$, we can efficiently construct a classical circuit $C_f$ comprising only Toffoli gates, which acts on the following computational basis states as follows:
\begin{equation}
 C_f (x ,1, 1_A) = (x,f(x),1_A),
\end{equation}
where $x\in\{0,1\}^n$ and $A$ is an ancilla register of size $O(n)$. 

Note that in both circuits $C_f$ and $A_f$, we do not allow for the addition of more ancilla lines or for the discarding of any bit or qubits. This is because for the notion of $\strong(n)$ simulation, all bit or qubit lines have to be accounted for. Hence, we make explicit the reference to the ancilla registers $A$. In this section, we present the details of the above constructions.

Recall the definition of a 3-CNF formula given in \eq{3cnf_formula}. As above, we assume that every variable $x_1,\ldots, x_n$ appears in the formula for $f$, so that $n\leq 3N$, i.e.\@ $n=O(N)$.

\subsection{Constructing $C_f$}

We show that we can implement the function $f$ using Toffoli gates alone. We denote the action of the Toffoli gate on lines $i,j,k$ with inputs $a,b,c$ by $$\Tof_{ijk}(\ldots, a,\ldots, b,\ldots, c,\ldots) = (\ldots, a,\ldots, b,\ldots, c\oplus a \cdot b, \ldots).$$
We use subscripts at the end to indicate a `marginalizing out' of the values of all other wires, for example,
$$\Tof_{143}(a,b,c,d,e)_{235} = (a,b,c\oplus a\cdot d,d,e)_{235} = (b,c\oplus a\cdot d,e).$$

\begin{lemma}
Let $f$ be a 3-CNF formula of the form given by \eq{3cnf_formula} with $n$ variables and $N$ clauses, where $n=O(N)$. Then there exists a classical circuit $C_f$ consisting of $O(N)$ Toffoli gates on $n+1+s(N)$ lines, for some $s(N) =O(N)$ (where we do not allow for the addition of bit lines or the discarding of any bits), such that
\begin{equation}
C_f (x_1,\ldots, x_n, 1, \underbrace{1, \ldots, 1}_{s(N)}) = (x_1,\ldots, x_n, f(x_1,\ldots,x_n), \underbrace{1, \ldots, 1}_{s(N)}).
\label{coff}
\end{equation}
\label{3CNFToffoli}
\end{lemma}

\begin{remark}
The ancilla bits are initialized to $1$ instead of $0$. This is because the Toffoli gate is universal only if we have the ability to prepare the state $1$. In particular, if the inputs were always just $0$'s, then it would not be possible to create the state $1$. On the other hand, we can prepare $0$ from $1$ since the target bit of $\Tof(1,1,1)$ is 0.
\end{remark}

\begin{proof}
We first show how to compute $f$ on the input $(x_1,\ldots, x_n)$ using AND, OR, NOT, COPY and SWAP gates . Let $k_i$ and $\bar k_i$ be the number of times $x_i$ and $\bar x_i$, respectively, appear as literals in the formula for $f$, i.e.\@ $\sum_i (k_i + \bar k_i)=3N$. By assumption, every variable $x_1,\ldots, x_n$ appears in the formula for $f$, so $k_i + \bar k_i >0$ for all $i$. 

For each $i$, if $k_i >0$, apply the COPY gate $k_i -1$ times to $x_i$ and the COPY gate followed by the NOT gate $\bar k_i$ times to $x_i$. Otherwise, if $k_i =0$ (i.e.\@ $\bar k_i >0$), apply the NOT gate followed by the COPY gate $\bar k_i$ times to $x_i$. This creates the state 
$$(\underbrace{x_1,\ldots, x_1}_{k_1},\ldots,\underbrace{x_n,\ldots, x_n}_{k_n}, \underbrace{\bar x_1,\ldots, \bar x_1}_{\bar k_1},\ldots,\underbrace{\bar x_n,\ldots, \bar x_n}_{\bar k_n}).$$
Note that the number of gates that the above procedure involves is $\sum_{k_i>0} \left[ (k_i-1)+2\bar k_i \right] + \sum_{k_i=0} 2 \bar k_i \leq 2\sum_i(k_i +\bar k_i) = 6N$. 

Applying the SWAP gate up to $3N$ times to the above state, we get the state 
\begin{equation}
\label{maxwidth}
(a_{11}, a_{12}, a_{13}, a_{21},\ldots, a_{N1}, a_{N2}, a_{N3}).
\end{equation}

We now apply the OR and AND gates according to the formula in \eq{3cnf_formula} to get $(a_{11} \vee a_{12} \vee  a_{13}) \wedge (a_{21} \vee a_{22} \vee  a_{23}) \wedge \ldots \wedge (a_{N1} \vee a_{N2} \vee  a_{N3})$. This involves a total of $2N$ OR gates and $N-1$ AND gates. Hence, the resulting circuit $B_f$, whose number of gates is bounded above by $6N + 3N + 2N + N-1 = O(N)$, computes:
$$B_f(x_1,\ldots, x_n) = f(x_1,\ldots,x_n).$$

 Note that the maximum width of $B_f$, which occurs when the state is given by \eqref{maxwidth}, is $3N$.

We now use the fact that the Toffoli gate together with the ability to prepare the ancilla state $1$ is universal for classical computing. In particular, they simulate the above gates as follows:
\begin{eqnarray}
\neg x &=&\Tof_{123}(1,1,x)_3, \nonumber \\
x \land y  &=& [\Tof_{123}\circ\Tof_{453}(x,y,1,1,1)]_3,\nonumber \\ 
\textrm{COPY}(x) &=& [\Tof_{123}\circ\Tof_{243}(x,1,1,1)]_{13}, \nonumber \\
x \lor y &=& [\Tof_{453}\circ \Tof_{123} \circ \Tof_{453} \circ \Tof_{342} \circ\Tof_{341} (x,y,1,1,1)]_3, \nonumber\\
\textrm{SWAP}(x,y) &=& [\Tof_{123}\circ \Tof_{321} \circ \Tof_{123} (x,1,y)]_{13}.
\label{toffoliUniv}
\end{eqnarray}

We append ancilla lines initialized to $1$ to $B_f$, and replace all the gates in $B_f$ by Toffoli gates according to the rules in \eq{toffoliUniv}, and apply additional swap gates (implemented by Toffoli gates) so that the first output of the circuit is $f(x_1,\ldots, x_n)$. Note that we do not discard any bits. Each of the replacements increases the number of ancilla lines by at most 3 and the number of gates by at most 4. Hence, both the total number of lines $a(N)$ and the number of Toffoli gates in the new circuit $B'_f$ are still $O(N)$. The action of $B'_f$ on the computational basis states is given by:
$$ B'_f (x_1,\ldots, x_n, \vec 1) = (f(x_1,\ldots, x_n), j_2, \ldots, j_{a(N)}), $$
where $(j_2,\ldots, j_{a(N)})$ are junk bits.

We now make use of the \textit{uncomputation trick} to reset the junk bits to $1$. Since the Toffoli gates are their own inverse, the inverse of $B'_f$ is obtained by applying the gates in $B'_f$ in the reverse order. Consider the circuit $B''_f$ that is formed as follows:  first apply  $B'_f$ to  $(x_1,\ldots, x_n, \vec 1)$. Introduce a new ancilla line, called $a$, initialized to $0$. Next, apply the CNOT gate $CX_{1a}$. Finally, apply $B'^{-1}_f$ to the first $a(N)$ bits to reset them back to $(x_1,\ldots, x_n, \vec 1)$. A circuit diagram for the above steps is shown in Figure \ref{fig:uncomputation}. This gives 
$$  B''_f (x_1,\ldots, x_n, \vec 1,0) = (x_1,\ldots, x_n, \vec 1,f(x_1,\ldots, x_n)).$$

\begin{figure}
\begin{align}
\Qcircuit @C=1em @R=1em {
\lstick{x_1} & & \multigate{5}{\hspace*{0.15cm} B'_f \hspace*{0.15cm} } & \ctrl{6} & \multigate{5}{B'^{-1}_f} & \qw & \rstick{x_1}\\ 
  & \vdots & &  &  & \vdots & \\
\lstick{x_n} & & \ghost{\hspace*{0.15cm} B'_f \hspace*{0.15cm}} & \qw & \ghost{B'^{-1}_f} & \qw & \rstick{x_n} \\
\lstick{1} & &\ghost{\hspace*{0.15cm} B'_f \hspace*{0.15cm}} & \qw & \ghost{B'^{-1}_f} & \qw & \rstick{1} \\
  & \vdots & & &  & \vdots & \\
\lstick{1} & & \ghost{\hspace*{0.15cm} B'_f \hspace*{0.15cm}} & \qw & \ghost{B'^{-1}_f} & \qw &\rstick{1} \\
\lstick{0} & & \qw & \targ & \qw & \qw & \rstick{f(x_1,\ldots,x_n)}\\
}
\nonumber
\end{align}
\caption{Uncomputation trick, in which the output bits, except for those in the target register, are reset to their input values. The state evolves as follows: $(x_1,\ldots x_n,1,\ldots,1,0) \rightarrow$  $(f(x_1,\ldots,x_n),j_2,\ldots, j_{a(N)},0) \rightarrow$ $ (f(x_1,\ldots,x_n),j_2,\ldots, j_{a(N)},f(x_1,\ldots,x_n))\rightarrow$ $(x_1,\ldots x_n,1,\ldots,1,f(x_1,\ldots,x_n))$. }
\label{fig:uncomputation}
\end{figure}
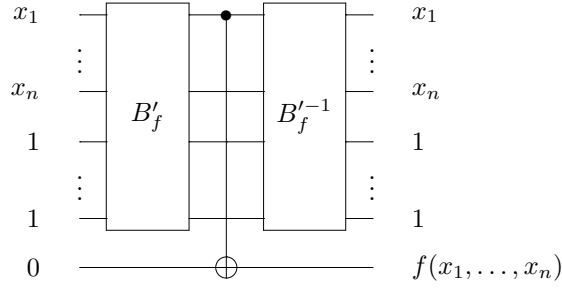

To get the required circuit $C_f$, we need to perform three more simple steps. First, the ancilla bit in the last register has to start from $1$ instead of 0. This can be achieved by applying a NOT gate (implemented by the Toffoli gate and ancillas initialized to 1) to 0. Second, the CNOT gate has to be simulated by a Toffoli gate. This may be achieved by using the fact that 
$$ CX(a,b)_{12} = \Tof_{123}(a,1,b)_{13}.$$
Third, the output has to be of the form \eqref{coff}. This is obtained by applying swap gates at the end of the circuit. These steps add at most a constant number of gates and a constant number of ancilla bits. Hence, the resulting circuit $C_f$ has $O(N)$ gates acting on $O(N)$ lines.

\end{proof}

\subsection{Constructing $Q_f$}

We now show how we can convert $C_f$ to a circuit $Q_f$ that involves only the basic Clifford gates and the T gate.

\begin{lemma}
Let $f$ be a 3-CNF formula of the form given by \eq{3cnf_formula} with $n$ variables and $N$ clauses, where $n=O(N)$. Then there exists a quantum circuit $Q_f$ consisting of $O(N)$ basic Clifford gates and $T$ gates on $n+1+s(N)$ lines (where we do not allow for the addition of qubit lines or the discarding of any qubits), such that
\begin{equation}
Q_f \ket{x_1,\ldots, x_n, 0, 0^{s(N)} } = \ket{ x_1,\ldots, x_n, f(x_1,\ldots,x_n), 0^{s(N)}},
\label{coffq}
\end{equation}
for some $s(N) =O(N)$. 
\label{3CNFCliffordT}
\end{lemma}

\begin{proof}

Using Lemma \ref{3CNFToffoli}, we have a circuit $C_f$ comprising $O(N)$ Toffoli gates satisfying
$$C_f \left(x_1,\ldots, x_n, 1, 1^{s(N)} \right) = \left(x_1,\ldots, x_n, f(x_1,\ldots,x_n), 1^{s(N)}\right).$$

Using the construction presented in \cite{NielsenChuang}, we express each Toffoli gate in terms of the basic Clifford gates, $T$ and $T^\dag$ gates, as follows:
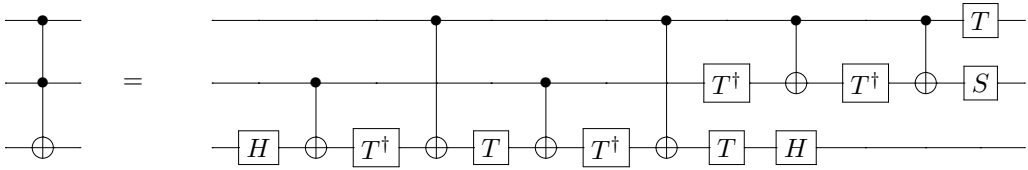
\begin{align}
\Qcircuit @C=1em @R=1em {
& \ctrl{1} &\qw & & & & &         &  \qw & \qw       & \qw                 & \ctrl{2}  & \qw & \qw & \qw                     &\ctrl{2}  & \qw &\ctrl{1} & \qw & \ctrl{1} &\gate{T} & \qw \\
& \ctrl{1} & \qw & & = & & &    & \qw  & \ctrl{1}   & \qw                 & \qw  & \qw     & \ctrl{1} & \qw                 &\qw     & \gate{T^\dag} & \targ & \gate{T^\dag} & \targ & \gate{S} & \qw \\ 
& \targ & \qw & & & &  &         & \gate H & \targ  & \gate{T^\dag} & \targ & \gate T & \targ & \gate{T^\dag}  & \targ & \gate T & \gate{H} & \qw & \qw & \qw &\qw
} 
\nonumber
\end{align}

Since $T^8 =1$, we replace each $T^\dag$ gate above by $T^7$. These replacements increase the number of gates by a constant factor, and hence the total number of gates in the new circuit is still $O(N)$. Finally, we insert $X$ (expressed as X = $HS^2H$) gates at the start and end of the circuit so that the ancilla lines start and terminate in the state $\ket 0$. This gives us a quantum circuit obeying \eq{coffq} with $O(N)$ wires and $O(N)$ gates.

\end{proof}

\end{document}